\documentclass[12pt]{amsart}

\usepackage{fullpage}
\usepackage{epsfig}
\usepackage{color}
\usepackage{extarrows}
\usepackage{amsfonts}
\usepackage{amssymb}
\usepackage{amsmath}
\usepackage{graphicx}

\usepackage{amsmath, epsfig, cite}
\usepackage{amssymb}
\usepackage{amsfonts}
\usepackage{latexsym}
\usepackage{longtable}
\usepackage{pgf}
\usepackage{tikz}
\newtheorem{theorem}{Theorem}[section]

\newtheorem{proposition}[theorem]{Proposition}
\newtheorem{corollary}[theorem]{Corollary}

\newtheorem{construction}[theorem]{Construction}

\newtheorem{thm}{Theorem}[section]

\newtheorem{prop}{Proposition}[section]
\newtheorem{cor}{Corollary}[section]
\newtheorem{rem}{Remark}[section]
\newtheorem{exm}{Example}[section]

\numberwithin{equation}{section}

\begin{document}

\title{Some Improvements on Locally Repairable Codes}
\author[J. Zhang]{Jun Zhang}
\address{School of Mathematical Sciences, Capital Normal University, Beijing 100048, P.R. China}
\email{junz@cnu.edu.cn}

\author[X. Wang]{Xin Wang}
\address{School of Mathematical Sciences, Zhejiang University, Hangzhou 310027, P.R. China}
\email{11235062@zju.edu.cn}
\author[G. Ge]{Gennian Ge}
\address{School of Mathematical Sciences, Capital Normal University, Beijing 100048, P.R. China}
\email{gnge@zju.edu.cn}
\thanks{The research of Gennian Ge was supported by the National Natural Science Foundation of China under Grant No. 61171198 and Grant No. 11431003, the
Importation and Development of High-Caliber Talents Project of Beijing Municipal Institutions, and Zhejiang Provincial Natural Science Foundation of China
under Grant No. LZ13A010001}
%\begin{abstract}
%In this section, we consider the
%\end{abstract}
%\subjclass[2010]{11S40; 11T23; 11L07}
%\keywords{locality, girth}
\maketitle

%\end{frontmatter}
 \begin{abstract}
The locally repairable codes (LRCs) were introduced to correct erasures efficiently in distributed storage systems. LRCs are extensively studied recently.
 In this paper, we first deal with the open case remained in \cite{q} and derive an improved upper bound for the minimum distances of LRCs. We also give an explicit construction for LRCs attaining this bound. Secondly, we consider the constructions of LRCs with any locality and availability which have high code rate and minimum distance as large as possible. We give a graphical model for LRCs. By using the deep results from graph theory, we construct a family of LRCs with any locality $r$ and availability $2$ with code rate $\frac{r-1}{r+1}$ and optimal minimum distance $O(\log n)$ where $n$ is the length of the code.
\end{abstract}

\section{Introduction}
In distributed storage systems, redundancy should be introduced to protect data against device failures. The simplest and most widespread technique used for data recovery is replication. However, this strategy entails large storage overhead and is nonadaptive for modern systems supporting the ``Big Data" environment. To improve the storage efficiency, erasure codes are employed, such as Windows Azure\cite{5}, Facebook's Hadoop cluster~\cite{14}, where the original data are divided into $k$ equal-sized fragments and then encoded into $n$ fragments $(n\ge k)$ stored in $n$ different nodes. It can tolerate up to $d-1$ node failures, where $d$ is the minimum distance of the erasure code. Particularly, the maximum distance separable (MDS) code is a kind of erasure code that attains the optimal
minimum distance with respect to the Singleton bound and thus provides the highest level of fault tolerance
for given storage overhead. However the MDS code is inefficient when we consider the disk I/O complexity, repair-bandwidth and so on.

To improve this, Gopalan et al.~\cite{3}, Oggier and Datta~\cite{7}, and Papailiopoulos et al.~\cite{10} introduced the concept of repair locality for erasure codes. The $i$th coordinate of a code has repair locality $r$ if it
can be recovered by accessing at most $r$ other coordinates. In this paper, an LRC is referred to an $[n,k]$ linear code
with all symbol locality $r$. When $r\ll k$, it greatly reduces the disk I/O complexity for repair.

Considering the fault tolerance level, the minimum distance is also a key metric for LRCs. Gopalan et al.~\cite{3} first derived the following upper bound for codes with information locality:
\begin{equation}\label{2}
d\le n-k+1-(\lceil \frac{k}{r}\rceil-1)
\end{equation}
which is a tight bound by the construction of pyramid codes~\cite{4}. Although the bound (\ref{2}) certainly holds
for all LRCs, it is not tight in many cases. Later, in \cite{9,2}, the bound (\ref{2}) was generalized to vector codes and nonlinear codes. %The results in \cite{3} pointed out that when $(r+1)\nmid n$ and $r\mid k$ the bound (\ref{2}) cannot be attained for codes with all symbol locality, and for those attaining this bound only the existence result was given for the case $(r+1)\mid n$ and the finite field needs to be large enough.
In order to consider multiple erasures in local repair, two different models were put forward independently by Prakash et al.~\cite{11} and Wang et al.~\cite{WZL15}.

For simplicity, the LRC that achieves the upper bound (\ref{2}) with equality is called an optimal (maximum) LRC in this paper. The first optimal LRCs for the case $(r + 1)|n$ were constructed explicitly in \cite{18} and \cite{15} by using Reed-Solomon codes and Gabidulin codes respectively. Both constructions were built over a finite field whose
size is exponential in the code length $n$. In \cite{19} for the same case $(r + 1)|n$ the authors constructed an optimal code over a finite field of size comparable to $n$ by using specially designed polynomials. This construction can be extended to the case $(r+1)\nmid n$ with the minimum distance $d\ge n-k-\lceil \frac{k}{r}\rceil+1$ which is at most one less than the upper bound~(\ref{2}). In~\cite{BargTV15,TamoBGC15,ZehY15}, the authors generalized this idea to the cyclic codes and algebraic geometry codes.

Recently, Song et al.\cite{16} carefully studied the  tightness of the bound (\ref{2}), and left two open cases. Another recent improvement was due to \cite{12} where Prakash et al. showed a new upper bound on the minimum distance for LRCs. This bound relies on a sequence of recursively defined parameters and is tighter than the bound (\ref{2}). But no general constructions attaining this new bound was presented. A great improvement for this problem is made by Wang and Zhang in\cite{q}. The authors carried out an in-depth study of the two problems: what is the largest possible minimum distance for an $[n,k]$ LRC? How to construct an $[n,k]$ LRC with the largest possible minimum distance? For the first problem, they derived an integer programming based upper bound on the minimum distance for LRCs, and then gave an explicit bound by solving the integer programming problem. The explicit bound applies all LRCs satisfying $n_1>n_2$, where $n_1=\lceil\frac{n}{r+1}\rceil$ and $n_2=n_1(r+1)-n$. For the second problem, they presented a construction of linear LRCs that attains the explicit bound for $n_1 > n_2$. Therefore, they had completely solved the two problems under the condition $n_1 > n_2$. A similar result can be found in~\cite{WFEH15} using matroid theory.

In this paper, we first deal with the open case remained in [40] and derive an improved upper bound for
the minimum distances of LRCs. We also give an explicit construction for LRCs attaining this bound.

There are lots of other works devoted to the locality in the handling of multiple node failures, such as \cite{13,17,19}, considering LRCs which permit parallel access of ``hot data", the works of \cite{8,20} studying LRCs with general local repair groups, and the work \cite{12} which proposed sequential local repair.
%In a word, more and more research work have concerned about codes with the local repair property, especially those codes attaining the largest possible minimum distance.
Very recently, Wang et al.~\cite{WZL15} proposed a binary LRC construction achieving any locality and availability with very high code rate.
An LRC code $C\, [n,k,d]$ is said to have locality $r$ and availability $t$, if for any codeword $y\in C$, any symbol $y_i$ of $y$
can be computed from some other $r$ symbols of $y$, and furthermore there are $t$ disjoint ways to reconstruct $y_i$. Unfortunately, the minimum distance of the codes constructed in~\cite{WZL15} is too small, saying $t+1$.

The second part of this paper deals with constructions of binary LRCs $C\, [n,k,d]$ with any locality $r$ and availability $2$ which have both high code rate and large minimum distance. We first give a graphical model for binary LRCs. We then use graphs with long girth to give a high rate code construction. Comparing with the constructions~\cite{WZL15,12}, our codes have a slight decline of rate, however, our codes have much larger minimum distance ($d=O(\log n)$).

This paper is organized as follows. Section~\ref{sec2} reviews some elementary results that will be used in this paper.
Section 3 solves the integer programming problem put forward in \cite{q}, and gives an explicit upper bound for LRCs satisfying $n_1\le n_2$. Then Section 4 presents an explicit construction attaining this bound. Section 5 gives a construction of a family of LRCs with any locality $r$ and availability $2$ having code rate $\frac{r-1}{r+1}$ and minimum distance $O(\log n)$ where $n$ is the length of the code.
Finally, Section 6 concludes the paper.

\section{Preliminaries}\label{sec2}
In~\cite{q}, the authors derived an integer programming based bound on the minimum distance of any LRC. Define
\begin{equation}\label{1}
\Psi(x)=\max_{s,t_1,\ldots,t_s \atop a_1,\ldots,a_s}\min_{l,h_1,\ldots,h_l}(xr+1-\sum_{i=1}^{l-1}(a_{h_i}-t_{h_i})), \forall 1\le x \le n_1,
\end{equation}
where $s,t_1,\ldots,t_s,a_1,\ldots,a_s$ satisfy
\[\left\{
\begin{array}{ccccccc}
t_1&+&\cdots&+&t_s&=&n_1;\\
a_1&+&\cdots&+&a_s&=&n_2;\\
a_i&\ge& t_i-1&,&\forall&1\le i\le s;\\
s\ge 1&;&t_i\ge 1&,&\forall& 1\le i\le s.
\end{array}
\right.
\]
and $l,h_1,\ldots,h_l$ satisfy
\begin{equation}\label{c}
t_{h_1}+\ldots+t_{h_{l-1}}<x\le t_{h_1}+\ldots+t_{h_l}.
\end{equation}

\begin{theorem}[\cite{q}]\label{5}
For any $[n,k,d]$ LRC,
\[
d\le n-k+1-\eta,
\]
where $\eta=\max\{x:\Psi(x)-x<k\}$.
\end{theorem}

Next, we review the construction of Tamo and Barg \cite{19} as their construction gives some optimal codes for the bound we will obtain later. Furthermore, we will employ their construction to get more optimal codes meeting our bound.

Let $A\subset F$, and let $\mathcal{A}$ be a partition of $A$ into $m$ subsets $A_i$. Consider the set of polynomials $F_{\mathcal{A}}[x]$ of degree
less than $|A|$ that are constant on the blocks of the partition:
\[
F_\mathcal{A}[x] = \{f\in F[x] : f \textrm{ is constant on } A_i , i = 1, \ldots ,m;\, \deg f < |A|\}.
\]
The annihilator of $A$ is the smallest-degree monic polynomial $h_A$ such that $h_A(a)=0$ if $a\in A$, i.e., $h_A(x)=\prod_{a\in A}(x-a)$. Observe that the set $F_\mathcal{A}[x]$ with the usual addition and multiplication modulo $h(x)$ becomes
a commutative algebra with identity. Since the polynomials $F_\mathcal{A}[x]$ are constant on the elements of $\mathcal{A}$, we write
$f(A_i)$ to refer to the value of the polynomial $f$ on the set $A_i\in \mathcal{A}$.

\begin{proposition}[\cite{19}]\label{3}
Let $\alpha_1,\cdots,\alpha_m$ be distinct nonzero elements of $F$, and let $g$ be the polynomial of degree $\deg(g)<|A|$ that satisfies $g(A_i)=\alpha_i$ for all $i=1,\cdots,m$, i.e.,
\[
g(x)=\sum_{i=1}^{m}\alpha_i\sum_{a\in A_i}\prod_{b\in A\backslash a}\frac{x-b}{a-b}.
\]
Then the polynomials $1,g,\cdots,g^{m-1}$ form a basis of $F_\mathcal{A}[x]$.
\end{proposition}

\begin{proposition}[\cite{19}]
There exist $m$ integers $0=d_0<d_1<\cdots<d_{m-1}<|A|$ such that the degree of each polynomial in $F_\mathcal{A}[x]$ is $d_i$ for some $i$.
\end{proposition}

\begin{corollary}[\cite{19}]
Assume that $d_1=r+1$, namely there exists a polynomial $g$ in $F_\mathcal{A}[x]$ of degree $r+1$, then $d_i=i(r+1)$ for all $i=0,\cdots,m-1$, and the polynomials $1,g,\cdots,g^{m-1}$ defined in Proposition~\ref{3}, form a basis for $F_\mathcal{A}[x]$.
\end{corollary}

\begin{construction}[\cite{19}]\label{4}
1. Let $F$ be a finite field, and let $A\subset F$ be a subset such that $|A|=n$, $n \bmod(r + 1)=s\neq0,1$. Assume also that $k+1$ is divisible by $r$ (this assumption is nonessential).

2. Let $\mathcal{A}$ be a partition of $A$ into $m$ subsets $A_1,\cdots,A_m$ such that $|A_i|=r+1,1\le i\le m-1$ and $1< |A_m|=s<r+1$. Let $g(x)$ be a polynomial of degree $r+1$, such that its powers $1,g,\cdots,g^{m-1}$ span the algebra $F_\mathcal{A}[x]$. W.L.O.G., assume that $g$ vanishes on the set $A_m$, otherwise one can take the powers of the polynomial $g(x)-g(A_m)$ as the basis for the algebra.

3. Let $a = (a_0, \cdots, a_{r-1})\in F^k$ be the input information vector, such that each $a_i$ for $i\neq s-1$ is a vector
of length $\frac{k+1}{r}$ and $a_{s-1}$ is of length $\frac{k+1}{r}-1$. Define the encoding polynomial
\[
f_a(x)=\sum_{i=0}^{s-2}\sum_{j=0}^{\frac{k+1}{r}-1}a_{i,j}g(x)^jx^i+\sum_{j=1}^{\frac{k+1}{r}-1}a_{s-1,j}g(x)^jx^{s-1}+\sum_{i=s}^{r-1}
\sum_{j=0}^{\frac{k+1}{r}-1}a_{i,j}g^j(x)x^{i-s}h_{A_m}(x).
\]
The code is defined as the set of evaluations of $f_a(x)$, $a\in F^k$.
\end{construction}

\begin{theorem}[\cite{19}]\label{Barg}
The code given by Construction~\ref{4} is an $[n,k,r]$ LRC code with minimum distance satisfying
\[
d\ge n-k-\lceil\frac{k}{r}\rceil+1.
\]
\end{theorem}

\section{Upper Bounds of the Minimum Distance}
In this section, we solve the integer programming problem~(\ref{5}), and derive an explicit upper bound for all LRCs satisfying $n_1\le n_2$. Then we make comparisons with the bound~(\ref{2}) to show the improvements of our explicit bound. Actually, in the next section we will show our bound is tight for the case $n_1\le n_2$.

\begin{theorem}\label{valueofpsi}
For $1\le x\le n_1$ and $n_1\le n_2$,
\[
\Psi(x)=xr+1.
\]
\end{theorem}

\begin{proof}
1. Set
\[ \left\{
\begin{array}{ccc}
s&=&1,\\
t_1&=&n_1,\\
a_1&=&n_2.
\end{array} \right.
\]
Then we have
\[
\Psi(x)\ge \min_{l,h_1,\cdots,h_l}(xr+1-\sum_{i=1}^{l-1}(a_{h_i}-t_{h_i}))=xr+1.
\]

2. Assume that for some $1\le x\le n_1$,
\[
\Psi(x)\ge xr+2.
\]
Then there exist integers $s$ and $t_i,a_i,1\le i\le s$, satisfying the constraints of the integer programming and
\[
\min_{l,h_1,\cdots,h_l}(xr+1-\sum_{i=1}^{l-1}(a_{h_i}-t_{h_i}))\ge xr+2.
\]
Therefore for all integers $l$ and $h_1,\cdots,h_l \in[s]$ satisfying the constraint (\ref{c}), we have
\begin{equation}\label{d}
\sum_{i=1}^{l-1}(a_{h_i}-t_{h_i})\le -1.
\end{equation}

 If there is some $i$ such that $t_i\ge x$, let $h_1=i$ in the constraint (\ref{c}), then
\[
\sum_{i=1}^{l-1}(a_{h_i}-t_{h_i})=0,
\]
which contradicts to (\ref{d}). And hence, the assumption $\Psi(x)\ge xr+2$ does not hold, and we finish the proof.

So we assume $t_i<x$, $\forall 1\le i\le s$. For $1\le i\le s$, define
\[
b_i=a_i-t_i.
\]
W.L.O.G, we assume that $b_1\le b_2 \le \cdots \le b_s$.
Since $t_i<x$, we can find $i_0=1,i_1,i_2,\cdots,i_p<s$ satisfying that
\begin{align*}
t_1+\cdots+t_{i_1-1}<&x\le t_1+\cdots+t_{i_1},\\
t_{i_1}+\cdots+t_{i_2-1}<&x\le t_{i_1}+\cdots+t_{i_2},\\
&\ldots\\
t_{i_{p-1}}+\cdots+t_{i_p-1}<&x\le t_{i_{p-1}}+\cdots+t_{i_p},\\
t_{i_p}+\cdots+t_s<&x.
\end{align*}
Then we have
\begin{eqnarray*}
\sum_{i=1}^{i_1-1}b_i\le-1,\\
\sum_{i=i_1}^{i_2-1}b_i\le-1,\\
\ldots,\\
\sum_{i=i_{p-1}}^{i_p-1}b_i\le-1.
\end{eqnarray*}
Noting $\sum_{i=1}^{s}b_i=n_2-n_1\ge 0$, we get $\sum_{i=i_p}^{s}b_i> 0$. Because $x\le n_1$, then $p\ge 1$, we can consider the last two parts of $t_i$'s in the reverse order $t_s,\cdots,t_{i_p},\cdots,t_{i_{p-1}}$.
From the definition of $i_p$, we know $\sum_{m=s}^{i_{p-1}}t_m\ge x$ and $\sum_{m=s}^{i_{p-1}}b_m\ge 0$.
So there exists $q$ satisfying $\sum_{m=s}^{i_{q-1}}t_m< x\le \sum_{m=s}^{i_{q}}t_m$ and $\sum_{m=s}^{i_{q-1}}b_m\le -1$. On the other hand $b_{i_{p-1}}\le b_{i_p} \le \cdots \le b_s$, we get $b_{i_{p-1}},\cdots ,b_{i_q}<0$, which contradicts to $\sum_{m=s}^{i_{p-1}}b_m\ge 0$.

Thus we have $\Psi(x)\le xr+1$.
\end{proof}

\begin{theorem}\label{t}
For any $[n,k,d]$ LRC with $n_1\le n_2$, where $n_1=\lceil\frac{n}{r+1}\rceil$ and $n_2=n_1(r+1)-n$, it holds that
\begin{equation}\label{t1}
d\le n-k+1-(\lceil\frac{k-1}{r-1}\rceil-1).
\end{equation}
\end{theorem}
\begin{proof}
This follows from Theorems~~\ref{5} and \ref{valueofpsi}.
\end{proof}

%\subsection{Comparison with Gopolan et al's Bound}
Since the bound (\ref{t1}) in Theorem~\ref{t} holds for $n_1\le n_2$, all the comparisons we make below are under the condition $n_1\le n_2$.

The bound (\ref{2}) given by Gopalan et al.\cite{3} is the first upper bound on the minimum distance of LRCs. Since $r< k$ (the natural condition that LRCs require), we always have $\lceil\frac{k-1}{r-1}\rceil \ge \lceil\frac{k}{r}\rceil$. So the bound (\ref{t1}) generally provides a tighter upper bound than the bound (\ref{2}).

Specially, we assume $k=ur+v$ for some integers $u,v$ and $0\leq v\leq r-1$, then
\[
d\le n-k+1-(\lceil\frac{k-1}{r-1}\rceil-1)=\left\{
                                             \begin{array}{ll}
                                               n-k-u+1, & \hbox{$u+v\le r$;} \\
                                               u-k-u, & \hbox{$u+v>r$.}
                                             \end{array}
                                           \right.
\]

\section{Code Construction When $n_1\le n_2$}
In this section, we present an explicit construction of LRCs attaining the bound (\ref{t1}) in some cases. The idea of construction comes from \cite{19}.

\begin{theorem}
When $n_1\le n_2$ and $u+v>r$, $n_2\neq r$, the bound (\ref{t1}) is achievable.
\end{theorem}
\begin{proof}
This follows from Theorems~\ref{Barg} and~\ref{valueofpsi}.
\end{proof}

Modifying the construction in the above theorem, we show that the bound~(\ref{t1}) is also tight in other cases.

\begin{construction}\label{const}
Let $F$ be a finite field, and let $A\subset F$ be a subset such that $|A|=n$.
\begin{enumerate}
                      \item Since $n=n_1(r+1)-n_2=n_1r-(n_2-n_1)$, let $\mathcal{A}$ be a partition of $A$ into $n_1$ subsets $A_1,\cdots,A_{n_1}$ such that $|A_i|=r,1\le i\le n_1-1$ and $|A_{n_1}|=s=r-(n_2-n_1)> 1$. Let $g(x)$ be a polynomial of degree $r$, such that its powers $1,g,\cdots,g^{n_1-1}$ span the algebra $F_{\mathcal{A}}[x]$. W.L.O.G., we assume that $g$ vanishes on the set $A_{n_1}$ and $u+v=s$ (this assumption is nonessential, in fact we only need $u+v\le s$).
                      \item Let $a=(a_0,\cdots,a_{r-2})\in F^k$ be the input information vector, such that $a_i$ is a vector of length $u+1$ for $0\le i\le s-1$ and $a_i$ is a vector of length $u$ for $i\ge s$. Define the encoding polynomial
\[
f_a(x)=\sum_{i=0}^{s-1}\sum_{j=0}^{u}a_{i,j}g(x)^jx^i+\sum_{i=s}^{r-1}\sum_{j=0}^{u-1}a_{i,j}g(x)^jx^{i-s}h_{A_{n_1}}(x),
\]
where $h_{A_{n_1}}=\prod_{a\in A_{n_1}}(x-a)$.
                    \end{enumerate}
The code is defined as the set of evaluations of $f_a(x)$, $a\in F^k$.

\end{construction}

\begin{theorem}\label{tightthm}
Keep the notation as above. The code given in Construction~\ref{const} is an $[n,k]$ LRC with locality $r-1$ and minimum distance
\[
d\ge n-k-u+1.
\]
\end{theorem}

\begin{proof}
Since the encoding is linear and the encoding polynomials have degree at most
\[
\max\{ur+s-1,(u-1)r+r-1\}=ur+u+v-1=k+u-1
\]
we have $d\ge n-deg(f)\ge n-k-u+1$.
The locality property is similar to Construction~(\ref{4}). If the erased symbol $f_a(x)$ lies in $x\in A_{n_1}$, by interpolating the other $s-1$ points in $A_{n_1}$  we get a polynomial of degree at most $s-2$ to recover $f_a(x)$. Otherwise, we use $r-1$ interpolation points to find a polynomial of degree at most $r-2$ to recover $f_a(x)$. So, it is an LRC code with locality ($r-1$). The result follows.
\end{proof}

As a corollary, we obtain more tight range for the bound~(\ref{t1}).
\begin{cor}
When $n_1< n_2$ and $u+v+n_2-n_1\le r$, the bound~(\ref{t1}) is achievable.
\end{cor}
\begin{proof}
We can view the code constructed above as an LRC with locality $r$, then the corollary follows directly from Theorems~\ref{t} and~\ref{tightthm}.
\end{proof}
\section{Graph-based Construction of LRCs with Arbitrary Locality and Availability $2$}
Very recently, Wang et al.~\cite{WZL15} proposed a binary LRC construction achieving any locality and availability with very high code rate.
In this section, we first give a graphical model for binary LRCs. Secondly, we consider the special case $t=2$, i.e., there are two
disjoint repair ways for any coordinate. We give a high rate code construction. Comparing with the construction~\cite{WZL15}, our
codes have a slight decline of rate, however, our codes have much larger minimum distances.

Recall that an LRC $C\, [n,k,d]$ with locality $r$ and availability $t$ satisfies the following property: for any codeword $y\in C$, any symbol $y_i$ of $y$
can be computed from some other $r$ symbols of $y$, and furthermore there are $t$ disjoint ways to reconstruct $y_i$.

\begin{prop}[\cite{17}]
For a linear code $C\, [n,k,d]$ with locality $r$ and availability $t$, the rate of the code satisfies
\[
 \frac{k}{n}\leq \prod_{i=1}^t \frac{1}{1+\frac{1}{ir}}.
\]
\end{prop}
The bound in the above proposition can not be achieved in most cases. Wang et al.~\cite{WZL15} gave a construction from the incidence matrices of some combinatorial designs:
\begin{prop}[\cite{WZL15}]
For any $r$ and $t$, there are binary linear codes $C\, [n,k,d]$ with locality $r$ and availability $t$ satisfying
\[
 \frac{k}{n}=\frac{r}{r+t}\qquad \textrm{and}\qquad d=t+1.
\]
\end{prop}
Note that for fixed $t$, the minimum distance of the construction in the above proposition is fixed. But as discussed at the beginning of this paper, the minimum distance of LRCs is a very important metric, especially for multiple erasures. So how to construct LRCs with high rate, large minimum distance, any locality and availability is the issue we care in the following content.

Next, we only consider the binary case. The method works as well for non-binary cases. To construct a binary LRC $C\, [n,k,d]$ with locality $r$ and availability $t$, it is equivalent to construct a parity check matrix $H$ such that each column has Hamming weight $\geq t$ and each row has Hamming weight $\leq r+1$ such that the inner product of any two rows is $1$.
Note that rows of $H$ might be linearly dependent. Corresponding to this parity check matrix
$$H=(h_{i,j})_{1\leq i\leq m,\,1\leq j\leq n},$$
 there is a bipartite graph $G$ whose bi-adjacent matrix is $H$. Explicitly, the graph $G=(V,E)$ is defined as following:
\begin{itemize}
  \item The set $V$ of vertices is separated into two parts $\{c_1,c_2,\cdots,c_m\}$ and $\{x_1,x_2,\cdots,x_n\}$ which represent rows and columns, respectively. The vertices $\{c_1,c_2,\cdots,c_m\}$ are always called constraints, and the vertices $\{x_1,x_2,\cdots,x_n\}$ are called variables.
  \item The set $E$ of edges: there are no edges connecting vertices in the same part, all the edges are connecting vertices from the distinct parts.  Precisely, there is an edge between $c_i$ and $x_j$ if and only if $h_{i,j}=1$ for any $1\leq i\leq m,\,1\leq j\leq n$.
\end{itemize}

\begin{exm}\label{exm}
For the matrix
\begin{equation*}
H=\left(
    \begin{array}{ccccccccc}
      1 & 0 & 0 & 1 & 0 & 0 & 1 & 0 & 0 \\
      0 & 1 & 0 & 0 & 1 & 0 & 0 & 1 & 0 \\
      0 & 0 & 1 & 0 & 0 & 1 & 0 & 0 & 1 \\
      1 & 0 & 0 & 0 & 1 & 0 & 0 & 0 & 1 \\
      0 & 1 & 0 & 0 & 0 & 1 & 1 & 0 & 0 \\
      0 & 0 & 1 & 1 & 0 & 0 & 0 & 1 & 0 \\
    \end{array}
  \right),
\end{equation*}
the corresponding graph is Figure 1. The matrix $H$ defines a $[9,4]$ binary code with locality $2$ and availability $2$.
\begin{figure}[h]\label{fig1}
\centering
\begin{tikzpicture}[scale=0.75]
     \tikzstyle{edge} = [draw,thick,-,black]
     \tikzstyle{point1}= [draw, circle, scale=0.6]
     \tikzstyle{point2}= [draw,scale=.8]

     \node[point1] (v1) at (-8,2) {$x_1$};
     \node[point1] (v2) at (-6,2) {$x_2$};
     \node[point1] (v3) at (-4,2) {$x_3$};
     \node[point1] (v4) at (-2,2) {$x_4$};
     \node[point1] (v5) at (0,2) {$x_5$};
     \node[point1] (v6) at (2,2) {$x_6$};
     \node[point1] (v7) at (4,2) {$x_7$};
     \node[point1] (v8) at (6,2) {$x_8$};
     \node[point1] (v9) at (8,2) {$x_9$};

     \node[point2] (v10) at (-5,-2) {$c_1$};
     \node[point2] (v11) at (-3,-2) {$c_2$};
     \node[point2] (v12) at (-1,-2) {$c_3$};
     \node[point2] (v13) at (1,-2) {$c_4$};
     \node[point2] (v14) at (3,-2) {$c_5$};
     \node[point2] (v15) at (5,-2) {$c_6$};

     \draw[edge] (v10) -- (v1) -- (v13) -- (v5) -- (v11) -- (v2) -- (v14) -- (v6) -- (v12) -- (v3) -- (v15) -- (v4) -- (v10) -- (v7) -- (v14);
     \draw[edge] (v11) -- (v8) -- (v15);
     \draw[edge] (v13) -- (v9) -- (v12);

\end{tikzpicture}
\caption{The bipartite graph $G$ representing $H$}
\end{figure}
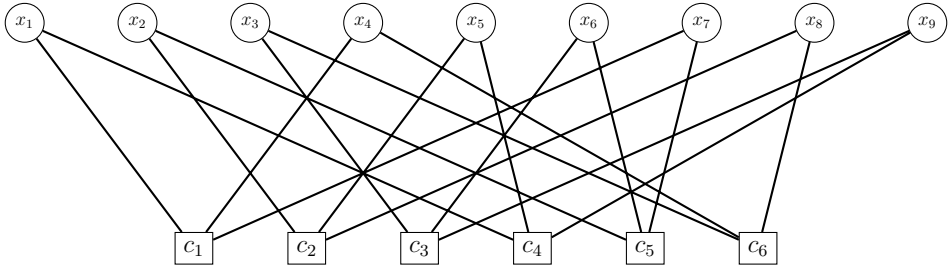

\end{exm}

In the literature, the bipartite graph is called Tanner graph. Recall that the degree of a vertex $\upsilon\in V$ is defined to be the number of edges connecting the vertex, and denoted by $\deg(\upsilon)$. In this setting, in order to construct a binary LRC $C\, [n,k,d]$ with locality $r$ and availability $t$, the problem is reduced to constructing a bipartite graph $G$ with vertices $\{c_1,c_2,\cdots,c_m\}\cup \{x_1,x_2,\cdots,x_n\}$ such that $\deg(c_i)\leq r+1$ and $\deg(x_j)\geq t$ for any $1\leq i\leq m,\,1\leq j\leq n$. Meanwhile, to simplify the discussion, we only consider the regular case. That is, the Tanner graph is a regular graph, where
$$\deg(c_1)=\deg(c_2)=\cdots=\deg(c_m)=r+1$$
 and
 $$ \deg(x_1)=\deg(x_2)=\cdots=\deg(x_n)=t.$$
In this case, we have the following lower bound for the corresponding code:
\begin{prop}
The rate $\rho$ of the binary code is
\[
 \rho\geq 1-\frac{t}{r+1}.
\]
\end{prop}
\begin{proof}
By counting the number of $1$'s in the bi-adjacent matrix $H$ of $G$, we have
\[
   m(r+1)=nt,\qquad\textrm{or}\qquad\frac{m}{n}=\frac{t}{r+1}.
\]
So the rate of the code is
\[
  \rho=1-\frac{\mathrm{Rank}(H)}{n}\geq 1-\frac{m}{n}=1-\frac{t}{r+1},
\]
where $\mathrm{Rank}(H)$ is the $\mathbb{F}_2$-rank of $H$.
\end{proof}
It is a very tough work to compute the exact value of $\mathrm{Rank}(H)$ in general. But it is an important issue in many application scenarios. For graphs with strong combinatorial property, computing $\mathrm{Rank}(H)$ attracts lots of interests~\cite{MacMann68,Smith69,BV92,EHKX99,CX03,JW04}.

There are advantages of graphical representation of codes~\cite{KFL98,Forney01}. It generalizes low-density parity-check codes, convolutional codes, trellis codes, classical linear system theory, behavior systems theory, etc. Fast algorithms on graphs give efficient encoding and decoding algorithms, such as the sum-product algorithm, BCJR algorithm, Viterbi algorithm, etc. In our specific case of LRCs, when the information of any node is not available or damaged, it is easy to recover the information by adding the information of neighboring variable vertices of any neighbor of the node in the graph. Even if many variable nodes are damaged, we can track in the graph for the intact information to recover the damaged nodes provided that the number of damaged nodes is less than the minimum distance of the code. This is our motivation to enlarge the minimum distance of the LRCs with the required locality and availability as large as possible.

Next, we restrict ourselves to the case $t=2$ where there are two disjoint repair options for each coordinate.
In other word, the degree of $x_j$ ($1\leq j\leq n$) is two in the Tanner graph. In this case, the rate of the code is
\[
    \geq\frac{r-1}{r+1}
\]
which might be smaller than that of~\cite{WZL15} by difference (at most)
\[
    \frac{r}{r+2}-\frac{r-1}{r+1}=\frac{2}{(r+1)(r+2)}.
\]
By slight sacrifice of the code rate, we can construct codes with much larger minimum distance. More concretely, the code in~\cite{WZL15} has
minimum distance $3$, but our codes have minimum distance $O(\log n)$.

Since $\deg(x_j)=2$ for all $1\leq j\leq n$, the Tanner graph $G$ can be reduced to a smaller graph $G_{red}$:
\begin{itemize}
  \item The vertices are $\{c_1,c_2,\cdots,c_m\}$.
  \item There is an edge between $c_i$ and $c_j$ if and only if $c_i$ and $c_j$ connect some $x_l$ simultaneously in the graph $G$.
\end{itemize}
The reduced graph $G_{red}$ is an ($r+1$)-regular graph. One could also refer the reduced graph $G_{red}$ as another graphical model of the code $C$. The difference between the two graphs is that Tanner graph considers constraints (or rows) as one part of the bipartite graph, but the reduced graph considers the constraints as the edges of the graph.

\begin{exm}
Continue Example~\ref{exm}. The reduced graph $G_{red}$ is Figure~2.
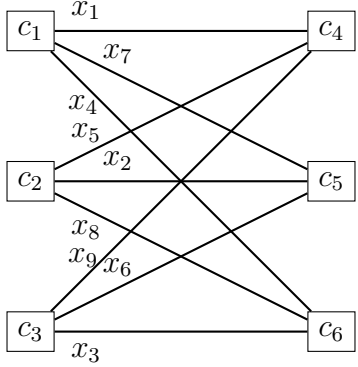
\begin{figure}[h]
\centering
\begin{tikzpicture}[scale=0.5]
     \tikzstyle{edge} = [draw,thick,-,black]
     \tikzstyle{point1}= [draw, circle, scale=0.8]
     \tikzstyle{point2}= [draw,scale=1]

     \node[point2] (v1) at (-4,4) {$c_1$};
     \node[point2] (v2) at (-4,0) {$c_2$};
     \node[point2] (v3) at (-4,-4) {$c_3$};
     \node[point2] (v4) at (4,4) {$c_4$};
     \node[point2] (v5) at (4,0) {$c_5$};
     \node[point2] (v6) at (4,-4) {$c_6$};

     \draw[edge] (v1) -- (v4) node[very near start, above] {$x_1$};
     \draw[edge] (v1) -- (v5) node[near start, above] {$x_7$};
     \draw[edge] (v1) -- (v6) node[very near start, below] {$x_4$};
     \draw[edge] (v2) -- (v4) node[very near start, above] {$x_5$};
     \draw[edge] (v2) -- (v5) node[near start, above] {$x_2$};
     \draw[edge] (v2) -- (v6) node[very near start, below] {$x_8$};
     \draw[edge] (v3) -- (v4) node[very near start, above] {$x_9$};
     \draw[edge] (v3) -- (v5) node[near start, above] {$x_6$};
     \draw[edge] (v3) -- (v6) node[very near start, below] {$x_3$};

\end{tikzpicture}
\caption{The reduced graph representing $H$}
\end{figure}
\end{exm}

To analyze the minimum distance of the code, we need one more index of the Tanner graph or the reduced graph. The girth of a graph is the length of a shortest
cycle in the graph. Since only graphs without multi-edges are involved in this paper, it is easy to see that the girth of a graph is $0$ or $\geq 3$, and the girth of a bipartite is an even integer: $0$ or $\geq 4$.
\begin{thm}\label{construction}
Let $G_{red}$ be an $(r+1)$-regular graph with $m$ vertices and girth $g$. Extend the graph $G_{red}$ to a bipartite graph $G$ with regularity $2$ and $r+1$. The null space of the bi-adjacent matrix $H$ of $G$ defines our binary linear code $C$. Then the code $C$ has length $\frac{m(r+1)}{2}$, dimension $\geq \frac{m(r-1)}{2}+1$, minimum distance $g$, locality $r$, and availability $2$.
\end{thm}
\begin{proof}
The only thing we need to prove is that minimum distance $d=g$. On one hand, W.L.O.G., let $c_1,c_2,\cdots,c_g$ be a cycle of length $g$ in $G_{red}$. Then it is extended to a cycle of length $2g$ in $G$, saying $c_1,x_1,c_2,x_2,\cdots,c_g,x_g$. By the fact $\deg(x_j)=2$, the restriction of the parity check matrix $H$ to the columns $x_1,x_2,\cdots,x_g$ is
\begin{equation*}
\begin{tabular}{|c|c|c|c|c|c|}
  \hline
  % after \\: \hline or \cline{col1-col2} \cline{col3-col4} ...
   & $x_1$ & $x_2$ & $x_3$&\ldots&$x_g$\\
   \hline
  $c_1$ & 1 & 1 &0&\ldots&0\\
  $c_2$ & 0 & 1 &1&\ldots&0\\
  $c_3$ & 0 & 0 &1&\ldots&0\\
  \vdots & \vdots & \vdots &\vdots&$\ddots$&\vdots\\
  $c_g$ & 1 & 0 &0&\ldots&1\\
  \vdots & 0 & 0 &0&\vdots&0\\
  \hline
\end{tabular}
\end{equation*}
which defines a codeword with support set $\{1,2,\cdots,g\}$. So the minimum distance $d\leq g$.

On the other hand, let $c$ be a codeword of Hamming weight $d$. W.L.O.G, assume the support of $c$ is $\{1,2,\cdots,d\}$. By non-zero location chasing, we prove $d\geq g$. We can also assume the variables $x_1,x_2$ are connected to the constraint $c_1$. Now, $x_2$ has the other neighboring constraint, saying $c_2$. As the codeword $c$ must satisfy the constraint $c_2$, there is at least one $x_j$ connecting $c_2$ for some $1\leq j\leq d$, $j\neq 2$. If $j=1$ then we get a cycle of length $4$ in $G$, so $g=2\leq d$ and the proof is finished. Otherwise, assume $x_3$ connects $c_2$. If $x_3$ also connects $c_1$, then the proof is finished as the same as the previous. Otherwise, $x_3$ connects the other constraint. Then iterate the same procedure. One can finally get a cycle of length $\leq 2d$ in the graph $G$. So the girth of the reduced graph $G_{red}$ is $g\leq d$.

In conclusion, we have proved the minimum distance $d=g$.

\end{proof}
The theorem extends the result of~\cite[Proposition~2]{HFE04}. For non-binary case, Chen et al.~\cite{CBSWJ} proposed how to enlarge the minimum distance by choosing proper non-zero elements at the non-zero locations of $H$. By Theorem~\ref{construction}, in order to construct a binary LRC with rate $\frac{r-1}{r+1}$, locality $r$, availability $2$, and minimum distance as large as possible, we need to construct an ($r+1$)-regular graph with girth $g$ as large as possible. This latter problem of graph construction has been extensively studied in extremal graph theory.

Let $g(m,r)$ denote the largest possible girth of an ($r+1$)-regular graph of size at most $m$, then for fixed $r$ and asymptotically growing $m$ we have
\begin{equation}\label{girth}
    (\frac{4}{3}-o(1))\log_r m\leq g(m,r)\leq (2+o(1))\log_r m.
\end{equation}
The second inequality in~(\ref{girth}) is a version of the Moore bound~\cite[Theorem III.1]{Extr04}. Note that the Moore bound is not achievable in most cases. The girths of random Cayley graphs are tested in~\cite{GHSSV09}. The first explicit construction can be found in~\cite{Mar82} for graphs with degree $4$ and large girth $\geq 0.83\log_r m$ and those with arbitrary degree and large girth $\geq 0.44\log_r m$, the latter of which was later improved in~\cite{Imrich84} to $\geq 0.48\log_r m$. Erd{\"o}s and Sachs~\cite{ES63} described a simple procedure yielding families of graphs with large girth $\log_r m$. Examples of graphs with arbitrary degree and large girth $\geq \frac{4}{3}\log_r m$ are given in~\cite{BH83,Weiss84,LPS88,Mar88,Mor94}. Using these explicit constructions, we can obtain
\begin{thm}
Let $G_{red}$ be an $(r+1)$-regular graph with $n$ edges and girth $g=O(\log n)$. Extend the graph $G_{red}$ to a bipartite graph $G$ with regularity $2$ and $r+1$. The null space of the bi-adjacent matrix $H$ of $G$ defines our binary linear code $C$. Then the code $C$ has length $n$, dimension $\geq \frac{n(r-1)}{r+1}+1$, minimum distance $O(\log n)$, locality $r$, and availability $2$.
\end{thm}
Comparing with the constructions in~\cite{WZL15,12}, our codes have a slight decline of rate, however, our codes have much larger minimum distances. Comparing with the construction of~\cite[Theorem 4.1]{19}, the minimum distances of their codes are very large apparently. On one hand, their construction relies highly on the size of the finite field, so their method can not be employed for the binary case. On the other hand, if their code rate achieves $\frac{r-1}{r+1}$, the minimum distance of their code degenerates to $1$. Actually, the minimum distance $O(\log n)$ in the above theorem is already optimal in the case $t=2$ by~\cite[Theorem~2.5]{Gallager63}.

\begin{rem}
Analogously to the performance of random linear codes, for general locality $r$ and availability $t\geq 3$, the codes constructed from random $(r+1,t)$-regular bipartite graphs have minimum distances with growing rate linearly to the length of the code with very high probability~\cite[Theorem~2.4]{Gallager63}. Within our knowledge, there is no deterministic construction for $(r+1,t)$-regular bipartite graphs (arbitrary $r$ and $t$) such that the corresponding codes have non-zero relative minimum distance $\frac{d}{n}$ asymptotically.
\end{rem}
\section{conclusions}
In the first part of this paper we studied the open problem in\cite{q}: when $n_1\le n_2$, what is the largest possible minimum distance for an $[n,k]$ LRC? How to construct an $[n,k]$ LRC with the largest possible minimum distance? For the first problem, we solve the linear integer programming in the case $n_1\le n_2$ and derive a new upper bound which is always better than the classic bound (\ref{2}). For the second problem, we find out that the construction of Tamo and Barg~\cite{19} is actually optimal when $n_1\le n_2$ and $u+v>r,v\neq r$. Using another interpolation polynomial, we present a construction of optimal LRCs when $n_1< n_2$ and $u+v+n_2-n_1\le r$.

In the second part of this paper, we presented a graphical model for binary LRC with any locality and any availability. In particular, for any locality and availability $2$, we use the deep results from extremal graph theory to give a code construction which produces good LRCs in the sense that these codes satisfy the locality and availability request and they have high code rates and large (indeed optimal) minimum distances.

\bibliographystyle{plain}
\bibliography{LRC}

%\begin{thebibliography}{00}
%\bibitem{Yang} {Minghui Yang}, {Jin Li}, {Keqin Feng} and {Dongdai Lin}, Generalized Hamming Weights of Irredubible Cyclic Codes, submitted to IEEE Trans. Inf. Theory, 2014.

%\end{thebibliography}

\end{document}